\documentclass[reqno]{amsart}

\usepackage{latexsym}
\usepackage{amsmath,amsthm}
\usepackage{amssymb}
\usepackage{graphicx}
\usepackage{hyperref}
\usepackage{tikz-cd}
\usepackage{mathtools}
\usepackage{amsfonts}

\newtheorem{theorem}{Theorem}[]
\newtheorem{lemma} [theorem]{Lemma}
\newtheorem{proposition} [theorem]{Proposition}

\newtheorem{corollary} [theorem]{Corollary}
\theoremstyle{definition}
\newtheorem{definition} [theorem]{Definition}
\newtheorem{rem}[theorem]{Remark}
\newtheorem{ex}[theorem]{Example}
\begin{document}

\title{A note on the distributions in quantum mechanical systems}
\author{Layth M. Alabdulsada}

\address{Institute of Mathematics, University of Debrecen, H-4002 Debrecen, P.O. Box 400, Hungary}
\email{layth.muhsin@science.unideb.hu}

\subjclass[2000]{58A30, 53C17, 81Q35, 93B05, 57R27}
\keywords{Distributions, affine distributions, Cartan decomposition, Lorentz group, quantum mechanical systems, controllability}

\bibliographystyle{alpha}
\begin{abstract}
In this paper, we study the distributions and the affine distributions of the quantum mechanical system.
Also, we discuss the controllability of the quantum mechanical system with the related question concerning the minimum time needed to steer a quantum system from a unitary evolution $U(0)=I$
of the unitary propagator to a desired unitary propagator $U_f$. Furthermore, the paper introduces a description of a $\mathfrak{k} \oplus \mathfrak{p}$ sub-Finsler manifold with its geodesics, which equivalents to the problem of driving the quantum mechanical system from
an arbitrary initial state $U(0)=I$ to the target state $U_f$, some illustrative examples are included.
We prove that the Lie group $G$ on a Finsler symmetric manifold $G/K$ can be decomposed into $KAK$.
\end{abstract}
\maketitle

\section{Introduction}
The research field that links geometric theories and quantum mechanics is yet very active research.
Geometric quantum mechanics is an attempt to reformulate quantum mechanics theory in the context of differential geometry. Differential geometry enters the quantum mechanics theory in at least two ways, the space of states and the space of time evolutions are the differential manifolds, see \cite{BH99} for an outstanding review of geometric quantum mechanics.
In the last two decades, there has been a resurgent interest  in the quantum mechanics systems with sub-Riemaniann (sub-Finsler) geometries. Excellent work can be found in \cite{ABB, JCG, KBG01, KGB02}. The objective of the controllability of quantum systems is to drive a dynamic system from an arbitrary initial state into a desired final state \cite{B14,Al07}, the interesting question about the controllability of quantum systems is steering the system in a minimum time \cite{KBG01, KGB02}.

 Let $M$ be a smooth $n$-dimensional manifold, and $k\leq n$. Suppose that for any $x \in M$, we specify a subspace of the tangent space of dimension $k$, i.e. $\mathcal{D} _{x} \subset T_{x}M$. Thus, for a neighbourhood $\mathcal{U}_{x}\subset M$ of $x$ there exist $k$ linearly independent smooth vector fields $X_{1}, \ldots, X_{k}$. Moreover, for each point $y \in \mathcal{U}_x,$
$$\mathcal{D}_{y} = \textrm{span} \{X_{1}(y), \ldots, X_{k}(y)\}.$$
Let $\mathcal{D}$ refer to the collection of all the linear $k$-dimensional subspaces  $\mathcal{D}_x$ for all  $x\in M$, then $\mathcal{D}$ is a {\it distribution} of rank $k$, or a $k$-{\it plane distribution} on $M$. The set of smooth vector fields $ \{X_{1}, \ldots, X_{k}\}$ is said to be a local basis of $\mathcal{D}$. Moreover, if now $X_{1}, \ldots, X_{k}, Y$ are  $k$ linearly independent (pointwise) vector
fields on $M$, and
the points of the distribution as a subset  $\mathcal{D}_x  \subset T_xM$ are given by
$$\mathcal{D}^+_x = Y(x)+ \mathcal{D}_x,$$
in other words,
$$\mathcal{D}^+_x = Y(x)+ \mathrm{span} \{X_1(x),X_2(x),\ldots,X_k(x)\}.$$
Then we label $\mathcal{D}^+_x$ as an  {\em affine distribution} of rank $k\leq n$, simply denoted by $\mathcal{D}_f$. The affine distribution $\mathcal{D}_f$ describes a control system since it catches all the possible directions of motion available at a definite point $x \in M$, up to a parametrization by control inputs (for more details, see \cite{PaSi02}).
In our recent work \cite{LA}, we studied in detail the distribution of the tangent and cotangent bundles. We also investigated the sub-Finslerian setting with nonholonomic mechanics. The study of affine distributions is more recent and less extensive, Clelland et al \cite{JCG} and Pappas et al \cite{PaSi02} provided a detailed exposition of affine distribution in the quantum mechanics systems.

This paper is organized as follows. In section 2, we have compiled some basic facts on the geometric description of the sub-Finsler manifolds and algebraic description of  Lie groups and Lie algebras in which we introduce some notions
which will play an essential role in this paper.
Next, the third section is devoted to the study of a homogeneous sub-Finsler manifold and a $\mathfrak{k} \oplus \mathfrak{p}$ sub-Finsler manifold.
Also, we show the decomposition of the Lie group $G$ as $KP$ if $G$ and its subgroup $K$ are connected such that $P= \exp (\mathfrak{p})$, and as $KAK$ on a Finsler symmetric manifold $G/K$.
Then, in section 4, we review in more general settings the quantum mechanics system with some examples. After that, section 5, is intended to motivate our investigation of the controllability of the quantum mechanical system in such a way we devoted to the study the controllability of quantum mechanical system in three cases: Firstly, the system has no drift hamiltonian.  Secondly,  a quantum system has a drift. Thirdly, the controllability of linear systems. It is also shown that each case mentioned above has a different distribution and how these distributions linked by the sub-Riemannian (sub-Finsler) geometries.

\section{Preliminaries}
Here, we define the geometric and algebraic notions that will be used throughout the paper.

\subsection{Geometric description of the sub-Finsler manifolds}

\begin{definition}\label{def}
A {\em sub-Finsler metric} is a generalisation of a sub-Riemannian metric on an $n$--dimensional smooth manifold $M$.
Let $F: \widetilde{\mathcal{D}}=\mathcal{D} \setminus\{0\}\rightarrow [0, +\infty)$ be a non-negative function on
a distribution $\mathcal{D}$ of rank $k \leq n$, which is a smooth
(vector) subbundle $\mathcal{D} \subset TM$ of the tangent bundle.

$F$  is  called  a sub-Finsler metric  in  $M$  if  it  satisfies  the  following  conditions:
\begin{itemize}

  \item [(1)] Smoothness: $F$ is $C^{\infty}$ on $\widetilde{\mathcal{D}}$;\\
   \item [(2)] Positive  homogeneity:  $F(\lambda v)=\lambda F(v)$ for all $v\in
    \widetilde{\mathcal{D}} \  \mbox{and}\ \lambda > 0$;\\
  \item [(3)] Strong  convexity: For  any  non-zero  vector  $v\neq 0$,  the  Hessian  matrix
formed  by  following
 $$
g_{ij} (x,v) = \frac{\partial^2F^2}{\partial v^i\partial v^j}(x,v), \qquad \mbox{for all}\ x \in M , v\in \widetilde{\mathcal{D}}_x,$$
is  positively  definite. Equivalently, the corresponding
indicatrix
$$ \Sigma_x=\{v\, | \, v \in \mathcal{D}_x, F(x,v)=1\}$$ is strictly convex.
\end{itemize}
A  differentiable manifold $M$ equipped with a sub-Finsler  metric  $F$ defined on a subbundle $\mathcal{D}$ of rank $k$ of a tangent bundle is called
a  {\em sub-Finsler  manifold}  or  sub-Finsler  space  denoted  by $(M, \mathcal{D}, F)$.
\end{definition}
The  Hessian  matrix for a sub-Finsler metric $F$ defines an inner product at the non-zero vector  $v\neq 0 \in \widetilde{\mathcal{D}}_x$, namely,
$$\langle u, w \rangle_v= g_{ij}(x,v)u^iw^j=\frac{\partial^2 F^2(v+tu+sw)}{\partial t \partial s}{\bigg|_{t,s=0}},$$
for any $u=u^i\frac{\partial}{\partial x^i}$ and $w=w^i\frac{\partial}{\partial x^i}$ in $\widetilde{\mathcal{D}}_x$,
the above inner product (bilinear form) on $\widetilde{\mathcal{D}}_x$ is called the {\em fundamental tensor}.
We note that if $\mathcal{D}= TM$, then this is the usual definition of a Finsler manifold \cite{BCS00}.

Let $(M, \mathcal{D}, F)$ be a sub-Finsler manifold and $x,  y \in M$. A piecewise smooth curve (trajectory) $\gamma :[0, T]\rightarrow M$ is said to be
\textit{horizontal}, or \textit{admissible} if $\dot{\gamma}(t)\in
\mathcal{D}_{\gamma(t)}$ for all $t\in [0, T]$, that is $\gamma(t)$ is
tangent to $\mathcal{D}$. The length of a piecewise
smooth horizontal curve $\gamma$ is
defined by
\begin{equation}\label{L1}
  \ell(\gamma)=\int_0^T F({\gamma(t)},\dot{\gamma}(t))dt.
\end{equation}
We can use the length to define the sub-Finslerian distance
$d(x,y)$ between two points $x,y \in M$ as in Finsler geometry:
$$ d(x,y)=\inf \{\ell(\gamma)| \  \gamma \ \mathrm{a\ piecewise\ smooth\
horizontal\ curve\ joining} \ x\ \mathrm{to}\ y\}.$$
 Furthermore, the sub-Finslerian distance is infinite, i.e. $ d(x,y)= \infty$  if there are no horizontal
curves between $x$ and $y$.

\begin{definition}
  The horizontal
curve $\gamma :[0, T]\rightarrow M$ is called a {\em length minimizing} (or simply a {\em minimizing}) geodesic if
it realizes the distance between its end points, that is, $\ell(\gamma) = d(\gamma (0), \gamma (T))$.
\end{definition}

Chow's theorem, also called the Chow-Rashevskii theorem, \cite{Chow39, Mo02}, answered a fundamental question of how we know these geodesic exist.
More precisely, given two points
$x$ and $y$ in a sub-Finslerian manifold, is there a geodesic joins them.
The answer for the above question depends whether the distribution is a bracket
generating or not.
To be more specific, if we have a bracket generating distribution on a connected
manifold $M$ then any two points in $M$ can be joined by a
horizontal path. A distribution $\mathcal{D}$ is called {\em bracket generating} if any local
frame $X_i$ of $\mathcal{D}$, together with all of its iterated Lie
brackets spans the whole tangent bundle $TM$, see \cite{Mo02}.

\subsection{Algebraic description}

Throughout this the paper, we assume $G$ a Lie group,
which is a connected group (smooth manifold) with the properties
that the map of the group multiplication and the inverse map are both smooth maps.
If $G$ is a Lie group, then $e \in  G$ will be the group identity element (we use
$I$ to denote the identity matrix when working with the matrix
representation of the group) and denote by $\mathfrak{g} = T_eG$ or Lie$(G)$
to be the Lie algebra of $G$.

Let $K$ be a (smooth submanifold) compact closed subgroup of a Lie group $G$.
Cartan's theorem asserts that every closed
subset $K$ of the Lie group $G$, which is its subgroup, is the Lie subgroup of the Lie
group $G$. Suppose Lie$(K)=\mathfrak{k}$ represent the
Lie algebra of $K$,
then the direct sum
decomposition of subalgebras $\mathfrak{k}$ and
$\mathfrak{p}$, denoted by
\begin{equation}\label{D}
  \mathfrak{g} = \mathfrak{k} \oplus \mathfrak{p}.
\end{equation}

We call the pair $({\mathfrak {g}},{\mathfrak {k}})$ {\em symmetric pair}
if there exists a Lie algebra
automorphism $\theta : G \to G$, such that $\theta^2 = I$,
where $\theta$ is the Cartan involution, and which has $\mathfrak {k}$ as its
$1$-eigenspace.
The direct sum decomposition (\ref{D}) well known as {\em Cartan decomposition} if
$\mathfrak{p}\subseteq \mathfrak{k}^{\bot}$ and $\mathfrak{k}\subseteq \mathfrak{p}^{\bot}$  with respect to the metric induced by the killing form  (\ref{Tr}). However,
if $\mathfrak{g}$ a semisimple Lie algebra, namely,  a Lie algebra $\mathfrak{g}$ is semisimple if
the only Abelian ideal in $\mathfrak{g}$ is \{0\} or
the killing form is non-degenerate on $\mathfrak{g}$,
then $\mathfrak{p}=\mathfrak{k}^{\bot}$ and $\mathfrak{k}=\mathfrak{p}^{\bot}$,
i.e. $\mathfrak{g}$ is the orthogonal sum of $\mathfrak{k}$ and $\mathfrak{p}$.
 Moreover, Cartan decomposition must satisfy the commutation relations, namely
 \begin{equation} \label{Cartan}
   [\mathfrak{k}, \mathfrak{k}] \subseteq \mathfrak{k}, \quad [\mathfrak{p}, \mathfrak{k}] \subseteq \mathfrak{p},  \quad [\mathfrak{p}, \mathfrak{p}] \subseteq \mathfrak{k}.
 \end{equation}
In addition, a pair $({\mathfrak {k}},{\mathfrak {p}})$ is called a {\em Cartan pair} of ${\mathfrak {g}}$.

 As an example for the Cartan decomposition one can take the Lorentz group $SO_0(n, 1)$,
 precisely, its Lie algebra $\mathfrak{so}(n, 1)$. The Lie algebra of the Lorentz group can be
 easily obtained by calculating the tangent vectors to curves $t \mapsto A(t)$, where $A(t)$ is a matrix, on $SO_0(n, 1)$
 through the identity $I$. Since $A(t)$ satisfies
 $$A^{\top}JA= J,  \qquad J= I_{n,1}= \begin{pmatrix} I_n & 0\\ 0 & -1 \end{pmatrix},$$
by differentiate and use the fact that $A(0) = I$, we obtain
 $$A'^{\top}J + JA' = 0.$$
 Consequently
 $$\mathfrak{so}(n, 1)= \{A \in M_{n+1}(\mathbb{R})|\ A^{\top}J + JA = 0\}.$$
 It follows that $JA$ is skew-symmetric accordingly to $J = J^{\top}$, and so
  $$\mathfrak{so}(n, 1)= \Bigg\{\begin{pmatrix} B & \lambda\\ \lambda^{\top} & 0 \end{pmatrix} \in M_{n+1}(\mathbb{R})\Bigg|\ \lambda \in \mathbb{R}^n, \ B^{\top}=-B \Bigg\}.$$
 It is appropriate to write $A^{\top}J + JA$ equivalent with $A^{\top}=-JAJ$ since $J^2=I$.
 Note that we can write each matrix $A \in \mathfrak{so}(n, 1)$  uniquely as
 $$\begin{pmatrix} B & \lambda\\ \lambda^{\top} & 0 \end{pmatrix}= \begin{pmatrix} B & 0\\ 0 & 0 \end{pmatrix} +\begin{pmatrix} 0 & \lambda\\ \lambda^{\top} & 0 \end{pmatrix},$$
 such that $\begin{pmatrix} B & 0\\ 0 & 0 \end{pmatrix}$ is skew-symmetric and $\begin{pmatrix} 0 & \lambda\\ \lambda^{\top} & 0 \end{pmatrix}$ is symmetric, further,
 both matrices are still in  $\mathfrak{so}(n, 1)$.
 Therefore, It is normal to identify that
 $${\mathfrak {k}}= \Bigg\{ \begin{pmatrix} B & 0\\ 0 & 0 \end{pmatrix}\Bigg|\ B \in M_{n}(\mathbb{R}),  \ B^{\top}=-B \Bigg\},$$
 and
  $${\mathfrak {p}}= \Bigg\{\begin{pmatrix} 0 & \lambda\\ \lambda^{\top} & 0 \end{pmatrix}\Bigg|\  \lambda \in \mathbb{R}^n \Bigg\}.$$
One can see at once clearly that both  $\mathfrak{k}$ and $\mathfrak{p}$ are subspaces (as vectors) of $\mathfrak{so}(n, 1)$.
However, on the one hand, $\mathfrak{k}$ is a Lie subalgebra isomorphic to $\mathfrak{so}(n)$,
on the other hand, since $\mathfrak{p}$ is not closed under the Lie bracket, so it is
not a Lie subalgebra of $\mathfrak{so}(n, 1)$.
 Nevertheless, (\ref{Cartan}) still valid.
 Therefore, the direct sum decomposition for the Lie algebra $\mathfrak{so}(n, 1)= \mathfrak{k} \oplus \mathfrak{p}$ is the Cartan decomposition.

 A {\em maximal Abelian subalgebra}, denoted by $\mathfrak{h}$, of the Lie subalgebra $\mathfrak{p}$ is
called a {\em Cartan subalgebra} of the symmetric pair $({\mathfrak {g}},{\mathfrak {k}})$, and the common dimension of all the maximal subalgebras is
called the {\em rank} of the decomposition (\ref{D}). As an example,
suppose that $\mathfrak{g}= \mathfrak{k} \oplus \mathfrak{p}$ is the Cartan decomposition
such that
$\mathfrak{g} = \mathfrak{su}(n)$, $\mathfrak{k} = \mathfrak{so}(n)$ and
 $$\mathfrak{p} = \{iS|S \ \mathrm{is}\ n \times n \ \mathrm{traceless \ real \ symmetric \ matrix}\},$$
 then $$\mathfrak{h} = \{iD|D \ \mathrm{is}\ n \times n \ \mathrm{traceless \ diagonal\ matrix}\}$$
 is a maximal abelian subalgebra contained in $\mathfrak{p}$, therefore
$$SU(n) = \{k_1 \exp(a)k_2|k_1, k_2 \in  SO(n), a \in\mathfrak{h}\}.$$
Let us turn to illustrate the action of Lie group $G$ on its Lie algebra $\mathfrak{g}$
by giving a brief exposition of its definition as follows:
 Let $G$ be a Lie group and $\mathfrak{g}$ its Lie algebra, the automorphism map on $G$
 $$\varphi_g : G \to G,$$
 given by the inner automorphism (conjugation)  $$ X \mapsto gXg^{-1}, \ \mathrm{for\ all} \ X \in \mathfrak{g}\ \mathrm{and} \ g \in G.$$
 In addition,  if we take the differential of $\varphi$ at the identity $e$ we get an
automorphism of the Lie algebra
 $$\mathrm{Ad}_g = ( \varphi_{g*})e : \mathfrak{g} \to \mathfrak{g}.$$
 Then for any Lie group G, the {\em adjoint representation} is defined by
 $$\mathrm{Ad} : G \to \mathrm{Aut}(\mathfrak{g}).$$

Now, let $G$ be a Lie group with Lie algebra ${\mathfrak {g}}$, then the exponential map defined by
$$\exp :{\mathfrak {g}}\rightarrow G.$$
 Given $g\in G, X\in \mathfrak{k}$, and consider the
one-parameter subgroup $ t \mapsto \exp tX$ of $K$
 \begin{equation*}
   \mathrm{Ad}_g(X)= gXg^{-1}\Big|_{t=0},
 \end{equation*}
 since $\exp (\mathrm{Ad}_{g}(X))= g (\exp X ) g^{-1},$  it follows that
 \begin{equation}\label{Ad}
   \mathrm{Ad}_g(\exp t X)= \frac{d}{dt} g(\exp t X)g^{-1}\Big|_{t=0},
 \end{equation}
 such that $XY-YX=[X,Y]=\mathrm{ad}(X)Y$ for any $X, Y \in \mathfrak{g}.$
 Therefore,
  \begin{equation}\label{ad}
  \mathrm{ad}_g(X)= \frac{d}{dt} \mathrm{Ad}_g(\exp t X),
  \end{equation}
where $\mathrm{ad}_g(X)$ is the {\em adjoint representation} of $\mathfrak{g}$.
The proof of the next theorem can be found in \cite{He01}, Chapter V.
\begin{theorem}\label{III}
  If $\mathfrak{h}$ and $\mathfrak{h'}$ are two maximal Abelian
subalgebras contained in $\mathfrak{p}$. Then
\begin{itemize}
  \item [(1)] There exists an element $X \in \mathfrak{h}$ whose centralizer in $\mathfrak{p}$ is equal to $\mathfrak{h}$.
  \item [(2)] There is an element $k \in  K $ such that $\mathrm{Ad}_k(\mathfrak{h}) = \mathfrak{h'}$.
  \item [(3)] $\mathfrak{p} =\bigcup_{k \in  K} \mathrm{Ad}_k(\mathfrak{h})$.
\end{itemize}
\end{theorem}
By the above theorem, the maximal Abelian subalgebras of $\mathfrak{p}$ are all conjugate
by $\mathrm{Ad}_k$.
Moreover, any two maximal flat totally geodesic submanifolds of $M$ must be congruent under the group
 of isometries of $M$ \cite{De12}. Also, if we assume $G$ a connected and compact Lie group such that $\mathfrak{g}$ its Lie algebra, then Theorem \ref{III}
 asserts that any two Cartan subalgebras of a complex semisimple Lie algebra are congruent under the group of automorphisms.
\begin{definition}
Let $G$ be a Lie group and $K$ be a compact closed subgroup of $G$,
the pair $(G,K)$ is called a {\em Klein geometry} where
the (left) coset space $G/K$ is connected.
For any Klein geometry $(G,K)$, the quotient space $M = G/K$ is a smooth manifold of dimension
$$\mathrm{dim} M = \mathrm{dim} G -\mathrm{dim} K.$$
In particular, there is a natural smooth left action of $G$ on $M$ such that the canonical mapping (projection) $\pi : G \to M$ is smooth, given by
$$L_g: G \times M \to M, \quad g_1 \cdot g_2 K:= (g_1g_2)K\ \mathrm{where} \ g_1, g_2 \in G,$$
 is transitive.
Moreover, $K$ acts on $G$ from the right by
$$\nu : G \times K \to G, \quad g\times k \mapsto gk, \quad g \in  G, k \in  K.$$
The quotient space $M $  can be done by the equivalence relation
$$g_1\sim g_2 \Leftrightarrow g_1=\nu(g_2, k), \quad \mathrm{for \ some} \ k \in K.$$
\end{definition}
The largest subgroup $\mathcal{K}$ of $K$ that is normal in $G$ is called the {\em kernel} of a Klein geometry.
The Klein geometry $(G, K)$ is called {\em locally effective} if the kernel is discrete
while is called an {\em effective} if the kernel is trivial, therefore, $K$ is compact.
In fact, $K$ is compact by the Ascoli-Arzel\`{a} theorem, as it is the stabilizer of the origin in $G/K$.
The kernel of a Klein geometry as defined above is well-defined, and
if we have a Klein geometry $(G,K)$ with kernel $\mathcal{K}$, then there is an associated effective Klein
geometry $(G/\mathcal{K},K/\mathcal{K})$ which gives the same smooth manifold, i.e. $$M= (G/\mathcal{K},K/\mathcal{K})\cong (G/K).$$
An effective Klein geometry is frequently  called a {\em homogeneous manifold}, i.e.
if $G$ a Lie group and $K \subset G$ a compact closed subgroup, we call the elements $g_1, g_2 \in G$
a congruent modulo $K$ if $g_1K=g_2K$ such that this is an equivalence relation to the equivalence classes being left cosets modulo $K$
and by this equivalence relation the quotient $(G/K)$ is called a {\em homogeneous space}.

Choose a $G$-invariant subbundle $\mathcal{D}$ (our hypothesis is always based on $\mathcal{D}$ is a bracket generating) of the tangent bundle $TM$ of $M=G/K$,
that is $\mathcal{D}$ of vector subspaces $\mathcal{D}_p\in T_pM$ for every $p\in M$ satisfying
$$(d\mu_{g})_{p}\mathcal{D}_p\subset \mathcal{D}_{\mu_{g(p)}}, \qquad \mathrm{for \ all} \ g \in G \ \mathrm{and }\ p \in M,$$
such that the diffeomorphism  map $\mu_g : M \to M$ defined by
the induced (isomorphism) tangent map  $(d\mu_{g})_{p} : T_pM \to T_{\mu_{g}(p)}M$ at $p \in M$ and $g \in G$.
 If the $G$-invariant subbundle $\mathcal{D}$ of
$TM$, where $M=G/K$ is a homogeneous space, then under this assumption
we can say that $\mathcal{D}$ is analytic (see, \cite{He01}).
Moreover, assume that $g \in G$, then the diffeomorphism given as
\begin{equation}\label{IS}
  \psi_{g}: G/K \to G/K, \ \ \  kK \longmapsto gkK.
\end{equation}
In general, $K$ will be called the
{\em isotropy subgroup}.

 The choice of the subbundle $\mathcal{D}\subset TM$ can be interpreted as follows. First of all,
 we can see an one-to-one- correspondence between $\mathrm{Ad}(K)$-invariant subspaces $V \in \mathrm{Lie}(G)$
 containing $\mathrm{Lie}(K)$, and $K$-invariant subspaces $\mathcal{D}_K \subset T_KM$.
Consider a subspace $\mathcal{D}_K \subset T_KM$, consequently, a
corresponding $V$ in $\mathrm{Lie}(G)$. Thus, the subbundle $\mathcal{D}$ is defined as
$$\mathcal{D}_{gK}:= (d\psi_{g})_K \mathcal{D}_K,$$
 for any $gK \in G/K$, where $\psi_{g}$ is given in (\ref{IS}).
Since $\mathcal{D}_K$ is $K$-invariant therefore the above expression does not depend on the representative in $gK$ which means the subbundle is well defined.
Let $V$ of $\mathrm{Lie}(G)$ be the subspace that is associated with $\mathcal{D}_K$ satisfies
the property that Lie$(G)$ is the smallest Lie subalgebra of Lie$(G)$ containing $V$, then
$V$ is called a {\em bracket generating}. Consider there exist a $G$-invariant norm on $\mathcal{D}$
if $K$ is compact, such that if we choose a seminorm on $V$ that is $\mathrm{Ad}(K)$-invariant with the property that
the kernel is $\mathfrak{k}$.
The natural projection $\pi:G \to G/K$ allows a $G$-invariant norm on $\mathcal{D}_K$ denoted by $\|\cdot\|$.
Therefore, we have an induced $G$-invariant norm on $\mathcal{D}$ as follows
$$\|v\|= \|(d\psi_{g^{-1}})_{gK}v\|, \qquad \mathrm{for \ all }\ v \in \mathcal{D}_{gK}.$$
The above equation is independent of the choice of the representative in $gK$ because the initial norm is $\mathrm{Ad}(K)$-invariant, see \cite{DoOt16}.

\section{Homogeneous sub-Finsler manifolds}
\begin{definition} \label{T6}
  The triple $(M, \mathcal{D}, F)$ is said to be a {\em homogeneous sub-Finsler manifold} where $M= G/K=\{gK:g\in G\}$ a smooth manifold
  such that $G/K$ is a homogeneous space formed by a Lie group $G$ modulo a compact subgroup $K$.
  Then we have a bracket generating $G$-invariant subbundle $\mathcal{D}\subset TM$, and $G$-invariant norm defined on $\mathcal{D}$.
  \end{definition}
  One can consider a sub-Finslerian metric by choosing a (semi) sub-Finslerian metric on a bracket generating subapsce $V$
  that is $\mathrm{Ad}(K)$-invariant and for which the kernel is $\mathfrak{k}$.
  We deduce that any sub-Finsler homogeneous manifold has the form $G/K$ with $G$ a Lie group and $K$ a compact subgroup, according to the fact that the isometry group of a homogeneous manifold is a Lie group.

  As in general case,
  we call the homogeneous sub-Finsler manifold $G/K$ {\em reductive}
  if there exists a subspace $\mathfrak{p}$ of the Lie algebra $\mathfrak{g}$ such that
  $$ \mathfrak{g}= \mathfrak{k} \oplus \mathfrak{p},\ \ \mathrm{where} \ \mathrm{Ad}(k)\mathfrak{p}\subset \mathfrak{p}, \forall k \in K,$$
  and we say that the above decomposition is a {\em reductive decomposition}.
  Further to the previous settings and their applications in the sub-Finslerian manifold. A $\mathfrak{k} \oplus \mathfrak{p}$ {\it sub-Finsler manifold} is a simple real Lie group $G$ of matrices with associated Lie algebra $\mathfrak{g}$
provided with such a sub-Finslerian structure.
For instance, on $G = SO_0(n, 1)$, consider the $G$-invariant subbundle $\mathcal{D} \subset \mathfrak{so}(n, 1)$ endowed with the $G$-invariant norm defined on $\mathcal{D}$.
In this case, we call the triple $(SO_0(n, 1), \mathcal{D}, \|\cdot\|)$ a $\mathfrak{k} \oplus \mathfrak{p}$ sub-Finsler manifold.

If $\mathfrak{k}$ satisfies the Cartan decomposition conditions, then it is a symmetric subalgebra. The corresponding
subgroup $K$ of $G$ is a symmetric subgroup and the coset space $G/K$ a symmetric
space. 
Furthermore, if we have an inner product, given by the (bilinear) killing form
  \begin{equation}\label{Tr}
    B(X, Y)= \mathrm{Tr}\big(\mathrm{ad}(X)\circ \mathrm{ad}(Y)\big),\quad X,Y \in  \mathfrak{g} \ \mathrm{and} \ \mathrm{ Tr \ is \ the \ trace },
  \end{equation}
  on $\mathcal{D}$ that is invariant under the action of $K$ since it is a compact, then the metric is a sub-Riemannian or a Riemannian (see for instance \cite{Mo02}). Consequently, we have a {\em Riemannian symmetric space}.

\begin{theorem}\cite{VB89} \label{Fin}
  Let $M$ be a Riemannian symmetric space, $G$ be the largest connected group of
isometrics of $M$, $K$ be the stabilizer of the group $G$ at the point $o \in M$. Then any $G$-invariant
intrinsic metric on $G/K$ will be Finsler.
\end{theorem}

\begin{proposition} \cite{PaSi02}
  An affine distribution $\mathcal{D}_f = X + \mathcal{D}$ on $M$ is invariant under a vector field $Z$
if and only if for every $Y \in  \mathcal{D}$:
$$[Z, Y] \subseteq Y.$$
\end{proposition}

\begin{lemma}\label{Le}
  Assume $K \subset G$ where $K$ is a subgroup of the Lie group $G$ with finite centre and their Lie algebra are $\mathfrak{k}$ and $\mathfrak{g}$, respectively. If $G$ and $K$ are connected, then $G=KP$, such that $P= \exp (\mathfrak{p})$.
\end{lemma}
\begin{proof}
  By using the action of the involutive automorphism $\theta: G \to G$,
  we can define $\mathfrak{k}$ and $\mathfrak{p}$ (and hence $K$ and $P$) as follows:
  For any $X \in \mathfrak{g}$  and consider the $+1$ and $-1$ eigenspaces of the derivative $d\theta: \mathfrak{g}\to \mathfrak{g}$,
  defined by $$ \mathfrak{k}=\{X \in \mathfrak{g} | \ d\theta(X)=X \},$$
  $$ \mathfrak{p}=\{X \in \mathfrak{g} | \ d\theta(X)=-X \},$$
   such that $\theta \neq I$ and $\theta^2 =I$.
   Applying the exponential map, these conditions become $\theta(k)=k$, $k \in K$. Further, for each $p \in P$
   $$\theta(P)= \theta(\exp \mathfrak{p})=  \exp\big(\theta(\mathfrak{p})\big)= \exp(-p)=p^{-1},$$
   on the involutive automorphism $\theta$ on $G$. Therefore, $G$ satisfies the decomposition $G= KP$.
\end{proof}
It is easy to see an automorphism on the Lie algebra $\mathfrak{so}(n, 1)$ of the Lorentz Group $SO_0(n, 1)$ that given earlier, as follows
 $$ \mathfrak{k}=\{A \in \mathfrak{so}(n, 1) | \ d\theta(A)=A \},$$
 and
  $$ \mathfrak{p}=\{A \in \mathfrak{so}(n, 1) | \ d\theta(A)=-A \},$$
  where the Cartan involution is $d\theta(A)=-A^{\top}=JAJ$ and involution $d\theta$ is the derivative at $I$ of the involutive isomorphism
   $\theta$ of
the group $SO_0(n, 1)$ also given by  $\theta(A)=JAJ$ and $A \in SO_0(n, 1)$.
\begin{theorem}
  Let $G/K$ be a Finsler symmetric manifold. If $\mathfrak{h}$ be a Cartan
subalgebra of the pair $(\mathfrak{g},\mathfrak{k})$ and define $A := exp(\mathfrak{h}) \subset G$. Then $G = KAK$.
\end{theorem}
\begin{proof} $G/K$ is a Finsler symmetric manifold by Theorem \ref{Fin},
in Lemma \ref{Le}, we decompose $G$ to $K$ and $P$. However, we can decompose $K$ and $P$ further.
Consider $(K)_0$ is the identity component of  $K$,
by the completeness in the Finslerian metric (for more details we refer the reader to \cite{BCS00}, Chapter 6), with the help of Lemma \ref{Le} implies that
$$P= \exp (\mathfrak{p}).$$ Moreover, use Theorem \ref{III} (III), we get
$$P= \exp \Big(\bigcup_{k \in  K} \mathrm{Ad}_k(\mathfrak{h})\Big)= \bigcup_{k \in  K} \mathrm{Ad}_k(\exp (\mathfrak{h}))=\bigcup_{k \in  K} \mathrm{Ad}_k\big(A\big) = \mathrm{Ad}_K\big(A\big),$$
for a fix cartan subalgebra $A$. By (\ref{Ad})
$P\subset KAK$, i.e. $p=k_1gk_1^{-1}$ for any $k_1 \in K$ and $g \in A$. Thus $G =KP= KKAK = KAK$.
\end{proof}
Observe that, the space $G/K$ induced by  a union of maximal Abelian subgroups $\mathrm{Ad}_k(A)$, called {\em maximal tori}, \cite{KBG01}.
The set
$$\mathrm{Ad}_K(X_d)=\{\mathrm{Ad}_{k_1}(X_d)=k_1X_dk^{-1}_1\big| k_1 \in K \}\in \mathfrak{p}$$
is called the {\em adjoint orbit} of $X_d$, which representee the directions in $G/K$, which we can choose
to move directly.

Assume we have the following decomposition $\mathfrak{g}= \mathfrak{k} \oplus \mathfrak{p}$
and  $\mathfrak{h} \subset \mathfrak{p}$ represent the maximal
Abelian subalgebra containing $X_d$.
Let $\mathfrak{W}_{X_d}= \mathfrak{h} \cap \mathrm{Ad}_K(X_d)$ denote the maximal
commuting set contained in the adjoint orbit of $X_d$, where the {\em Weyl
orbit} of $X_d$ is the set of $\mathfrak{W}_{X_d}$.
We define the
convex hull of the Weyl orbit of ${X_d}$
with vertices given by the elements of the Weyl
orbit of ${X_d}$ as follows: $$\mathfrak{c}(X_d)= \{ \sum_{i=1}^{n} \beta_i X_i\big| \beta_i\geq 0, \  \sum \beta_i=1,\ X_i \in\mathfrak{W}_{X_d}\}.$$

\begin{theorem}
  If $\mathfrak{h} \subset \mathfrak{p} $ and $\Gamma: \mathfrak{p} \to \mathfrak{h}$
  is the orthogonal projection with respect to the
$G$-invariant norm. Then for any $X \in \mathfrak{h}$, we have
$\Gamma(\mathrm{Ad}_KX) = \mathfrak{c}(W\cdot X)$,
where $\mathfrak{c}$ denotes convex hull.
\end{theorem}
\begin{proof}
  The proof is similar to Kostant's proof \cite{Ko73}, so we omit it.
\end{proof}

\section{Quantum mechanical systems}

In general, the physical systems are dynamic, i.e. they evolve over time.
In non-relativistic quantum mechanics, the time evolution
of a quantum system (e.g., an atom, a molecule, or a system of
particles with spin) can be described by a map $\psi : \mathbb{R} \to \mathbb{S}$, where
the domain $\mathbb{R}$ is the set of all real numbers and the range $\mathbb{S}$
a unit sphere in a complex separable Hilbert $\mathcal{H}$,
which is a complex valued function of the real variable called a {\em wavefunction} \cite{Do}.
The starting point is to let $\psi(0)$ be the initial state of the system, and $\psi(t)$ be the state vector at some
other time $t$. Now, the unitary state evolution of a quantum system
is given by
$$|\psi(t)\rangle= U(t)|\psi(0)\rangle,$$
where $U(t)$ unitary  (propagator) transformation.

In order to implement a certain quantum information processing task in a controlled quantum system, we consider the dynamics of the system, i.e. the time evolution operator can be defined through the time-dependent
Schr\"{o}dinger equation
\begin{equation}\label{T1}
  \frac{dU}{dt}(t) = - iH(t)U(t), \quad U(0)=I,
\end{equation}
where $H(t)$ is a (Hamiltonian) self-adjoint operator acting on a complex separable Hilbert space.
We can separate the Hamiltonian as
\begin{equation}\label{T2}
  H= H_d+ \sum_{j=1}^{m} u_jH_j,\ \qquad \ m < n,
\end{equation}
such that $H_d$ is the internal Hamiltonian part of the system
and it is called the {\em drift} or {\em free Hamiltonian} and the
second part of the Hamiltonian $\sum_{j=1}^{m} u_jH_j$ represent coherent manipulations from outside
called the
{\em control Hamiltonian}, where $u_j\in \mathbb{R}$ are controls that can be switched on and off,  see \cite{Al07}.
Moreover, $H_d$ and $H_j$ are traceless Hermitian matrices, see the example below.
\begin{ex}\label{in}
Consider the group of
unitary unimodular $2 \times 2$ complex matrices
  $$SU(2)= \Bigg\{\begin{pmatrix} \mu & \nu\\ -\bar{\nu} & \bar{\mu} \end{pmatrix} \in M_{2}(\mathbb{C})\Big|\ |\mu |^2 + |\nu|^2 =1 \Bigg\},$$
it is compact and simply connected.
 Its Lie algebra $\mathfrak{g}= \mathfrak{su}(2)$  of antiHermitian traceless $2 \times 2$ complex matrices described by
  $$\mathfrak{su}(2)= \Bigg\{\begin{pmatrix} i\mu & \nu\\ -\bar{\nu} & -i\bar{\mu} \end{pmatrix} \in M_{2}(\mathbb{C})\Big|\ \mu \in \mathbb{R},\ \nu \in  \mathbb{C} \Bigg\}.$$
The generators (basis) of $\mathfrak{su}(2)$ are $\{-i\sigma_x, -i\sigma_y, -i\sigma_z\}$ where
$\sigma_x, \sigma_y,$ and  $\sigma_z$ are matrices called Pauli matrices, given by
\begin{equation}\label{T8}
\begin{aligned}
\sigma_x ={} & \frac{1}{2} \begin{pmatrix} 0 & 1\\ -1 & 0 \end{pmatrix};
\sigma_y = \frac{i}{2} \begin{pmatrix} 0 & 1\\ 1 & 0 \end{pmatrix};
\sigma_z =\frac{i}{2} \begin{pmatrix} 1 & 0\\ 0 & -1 \end{pmatrix},
\end{aligned}
\end{equation}

they are traceless Hermitian, where $i$ is the imaginary unit, the matrices $\sigma_x, \sigma_y,$ and  $\sigma_z$
satisfy the commutation relations, namely,
$$[\sigma_x, \sigma_y]= \sigma_z, \ [\sigma_y, \sigma_z]= \sigma_x, \ [\sigma_z, \sigma_x]= \sigma_y$$
which completely describe the Lie algebra $\mathfrak{su}(2)$,
the matrix commutator defined as before (\ref{ad}), specifically $[X,Y]= XY-YX$.
The choice of the subspaces
$$\mathfrak{k}=\mathrm{span} \{\sigma_z\}, \qquad \quad \qquad  \mathfrak{p}=\mathrm{span} \{\sigma_x, \sigma_y\}$$
allows a Cartan decomposition for $\mathfrak{su}(2)$. Furthermore,
$\{\sigma_x, \sigma_y\}$ is an orthonormal frame, with respect to the Riemannian metric $g$ \cite{L.K1}, for the inner product (see Definition \ref{def})
restricted to $\mathfrak{p}$.
Now, on $G = SU(2)$, define the $G$-invariant subbundle $\mathcal{D}_x = x\mathfrak{p}$
endowed with the $G$-invariant norm defined on it.
Therefore, the triple $(SU(2), \mathcal{D}, \|\cdot\|)$ is a $\mathfrak{k} \oplus \mathfrak{p}$ sub-Finsler manifold.
\end{ex}

Turn back to the quantum system,
 Khaneja et al \cite{KBG01, KGB02} calculated the minimum
time it takes to steer this system (\ref{T1}) from a unitary evolution $U(0)=I$
of the unitary propagator to a desired unitary propagator $U_f$.
In other words,
if the Hamiltonian does not change in time, then the time evolution operator for time $t$
is the unitary operator $U(t)$.
Note that if we have  a bounded unitary operator $U(t)$, there exists a unique unitary operator $U^{\dag}(t)$ called the {\em adjoint} of $U(t)$ acting on $\mathcal{H}$.
However, a unitary operator satisfies the relation $U(t)U^{\dag}(t)=1.$

The problem we are ultimately interested in is to find the
minimum time required to transfer the density matrix from
the initial state $\rho(0)$ to a final state $\rho_f$. Thus, we will be
interested in computing the minimum time required to steer
the system
\begin{equation}\label{T4}
   \frac{dU}{dt}(t)= - i\left(H_d+ \sum_{j=1}^{m} u_j(t)H_j,\right)U(t),
\end{equation}
from identity, $U(0)=I$, to a final propagator $U_f$, this system defines a
bilinear control system, as it is linear both in the unitary operator  $U(t)$ and in the
control amplitudes $u_j(t) \in \mathbb{R}$. Observe that, if you have a finite $H_d$, then it will
take infinite time to get very far place, however, if you fix a final terminal point
then there will be a finite time to get there.
System (\ref{T4}) is exactly (respectively, approximately) controllable if every point of $\mathbb{S}$  can be
steered to (respectively, steered arbitrarily close to)
any other point of $\mathbb{S}$, by a horizontal curve of (\ref{T4}).

For a finite dimensional quantum mechanical system, if $n$ is the number of energy levels we have
$\mathcal{H} = \mathbb{C}^n$ and the state space $\mathbb{S}$ is the unit sphere  $\mathbb{S}^{2n-1} \subset \mathbb{C}^n.$
In this setting,  one can naturally associated with (\ref{T1}) (and consequently (\ref{T4})) is its lift on the
unitary group $SU(n)$, such that the solution is of the form
$$U(t)= g(t)U(I), \quad \ \mathrm{ with}\ g(t) \in SU(n).$$
Moreover, we can write (\ref{T1}) as
\begin{equation}\label{T5}
  \frac{d}{dt}g(t) = - iH(t)g(t), \quad U(0)=I,
\end{equation}
where $-iH(t)$ is a traceless skew
Hermitian matrix, that belongs to the Lie algebra $\mathfrak{su}(n)$ \cite{B14}.
The problem of controllability is understandable nowadays,
more precisely, the typical problem one is indeed interested in a quantum control
is that one wants to steer the state of a quantumical system from a given initial to a target state for each couple of points in $SU(n)$.
Furthermore,  we can actuate or manipulate the system by modulating the hamiltonian as a function of time.
In fact, the system is controllable if and only if the H\"{o}rmander's condition satisfies
$$\mathrm{Lie}\{iH_0, iH_1, \ldots , iH_m\} = \mathfrak{su}(n).$$
Let us consider the following
example of the subalgebra is $\mathfrak{su}(2)$ which is spanned by the multiples of the Pauli matrices (\ref{T8}).
Recall that the system Lie algebra $H$ of (\ref{T4}) is
defined as the smallest Lie subalgebra of $\mathfrak{g}$ containing $H_d, H_1,\ldots,H_m$, i.e. $H$ is the smallest
linear subspace $\{H_d, H_1,\ldots,H_m\}$ of  $\mathfrak{g}$, that contains $H_d, H_1,\ldots,H_m$ together with all the
iterated Lie brackets $$[H_d, H_i], [H_i, H_j ], [H_d, [H_i, H_j ]],\ldots.$$ It is easily seen, that the associated
Lie group $\exp(H)$ is the smallest subgroup of $G$ that contains $\mathbb{S}$, i.e. $\exp(H)$ is equal to the
system group $G$.

\section{Controllability in a quantum system}

Controllability is a major matter in system analysis before a control strategy is applied, and is
used to judge whether it is possible to control or stabilize the system.
In other words, it is the ability to steer a quantum system from a given
initial state to any final state, infinite time, using the
available controls. The controllability of the system (\ref{T4}) depends on some conditions that we will clarify in detail during this section.

First of all, the case of the quantum system that referred to earlier in (\ref{T4}) which has no drift hamiltonian $H_d$ instead it is just evolution in the unitary group
that can be switched on and off certain generators or hamiltonian $H_j$, such that the form of the system is as follows
\begin{equation}\label{T41}
   \frac{dU}{dt}= - i\left( \sum_{j=1}^{m} u_jH_j,\right)U \in \mathrm{span}\{iH_1U, \ldots , iH_mU\},
\end{equation}
then if these generators were rich enough to span the tangent space at each point of the unitary group,
then basically the problem of the controllability will be trivial which is better than we can steer
the system to the desired target point or not.
Now, if we assume that our unitary group is connected, then we can take two points and connect them by
a curve and follow this curve, more precisely, follow its tangent by choosing generators appropriately
with the required way.
The problem will become more interesting when this number of generators $H_j$ that we have
is actually much smaller than the tangent space of the unitary group, however, in this case
we can switch on and off these hamiltonian, one can produce a commutator or a new generator of these hamiltonian
as in the following procedure
\begin{equation}\label{NG}
\begin{aligned}
U(\delta t)= {} & \exp(iH_2\delta t)\exp(iH_1\delta t)\exp(-iH_2\delta t)\exp(-iH_1\delta t)\\
= {} & -(\delta t)^2 {[iH_1,iH_2]} \approx I- (\delta t)^2 \underbrace{[iH_1,iH_2]}_{\text{new generator}},
\end{aligned}
\end{equation}
such that the backward evolution $\exp(-iH_{1,2}\delta t)$
 is generated by letting the forward map $\exp(iH_{1,2}\delta t)$ evolve for sufficient period of time,
which might not lie in the span of these origin generators that one had at one's disposal,
which gives us a new direction of motion. Therefore, we can gen-generalized by taking such all iterate commutators
and generates the lie algebra such that if these all iterate commutators span the tangent space (the Lie algebra)
for the unitary group then one can show the system is controllable, these results goes by Chow's theorem \cite{Chow39}.

Concerting  $u_j$, note that we have the freedom to choose them positive and negative, i.e.
they can be made arbitrary large. So, if the control amplitude unbounded, then the point that can be reached, it can be reached with no time under our considering that $u_j$ was chosen large enough.
Consequently, if the system (\ref{T41}) is controllable then its controllable in arbitrary small amount of time.
To sum up the controllability of the system (\ref{T4}), if the Lie Algebra $\{-iH_j\}_{\mathfrak{g}}$ generated by
$\{-iH_j\}$ span the Lie algebra of the unitary group, then the system (\ref{T4}) is controllable.
The controllability of a driftless control system can described as a nonholonomic control system,
for more details and examples inspired by robotics, we refer the reader to \cite{LA}.

Second case, the controllability of a quantum system has a drift, i.e $H_0$,
that mostly takes the following form
\begin{equation}\label{T10}
   \frac{dU}{dt}= - i\left(H_0+ \sum_{j=1}^{m} u_jH_j,\right)U \in -iH_0U + \mathrm{span}\{iH_1U,\ldots, iH_mU\} ,
\end{equation}
if the Lie algebra $\{-iH_0, -iH_j\}_{\mathfrak{g}}$ span the Lie algebra of the unitary group, then the system is controllable.
However, one can fundamentally generate these commutators between $H_0$ and $H_j$ by getting
evolution negative in a direction of $H_0$ by forwarding enough evolution since the group
being compact.
 If the Lie algebra
span by $H_j$ is not a full unitary group, then the system is controllable this is due to the fact that we can not change the strength of $H_0$.
Keep in mind, no matter how long your control is, there is a minimum time to reach
anywhere, that time can not be shrunk down to be zero.
Regarding the system (\ref{T10}), one of the assumption to make the time is minimum for going from
the identity to the desired state by making the control $u_j$ arbitrary large in addition to a good
approximation in the setting when the strength of control Hamiltonian can be made much larger than
the drift Hamiltonian in the system. If $G$ is the unitary matrix group $SU(2^n)$. Then
the problem of finding the most efficient  radio-frequency (rf) pulse train required to evolve the system
to the desired state is therefore equivalent to the problem of finding a time optimal curve from $U(0) = I$ to the desired $U_f$. We can generate a subgroup $K$ from the control Hamiltonian $H_j$ defined by
 \begin{equation}\label{K}
  K\subset SU(2^n)=\exp(\{H_j\}_{\mathfrak{g}}),
 \end{equation}
where $\{H_j\}_{\mathfrak{g}}=\mathfrak{k}$ is the Lie algebra generated by elements $\{iH_0, iH_1, \ldots , iH_m\}$,
the subgroup $K$ is a collection of unitary propagators which can be generated
as long as the drift Hamiltonian $H_d$ is removed from the system (\ref{T4}).
We suppose the strength of
Hamiltonian controls could be arbitrarily made large.

If we assume that $\mathfrak{g}$ represent the Lie algebra of unitary group and $\mathfrak{k}$  be a subalgebra of $\mathfrak{g}$, which
correspond to the generator that our control Hamiltonian can be produced.
Then the problem of finding a minimum time to produce a propagator can be characterized this  minimum time in special settings
in which if we have the Cartan's decomposition $$\mathfrak{g}= \mathfrak{k} \oplus \mathfrak{p}$$ and this decomposition
satisfy the commutation relations, specifically,
$$[\mathfrak{k}, \mathfrak{k}] \subseteq \mathfrak{k}, \quad [\mathfrak{p}, \mathfrak{k}] \subseteq \mathfrak{p},  \quad [\mathfrak{p}, \mathfrak{p}] \subseteq \mathfrak{k}.$$

So, we can define an inner product by the natural killing form such that
$\mathfrak{p}$ and  $\mathfrak{k}$ are orthogonal to each other. Consequently, one has a
sub-Riemannian metric or sub-Finsler metric, see Example \ref{in}; additionally, the example below shows the expression of geodesics.

 \begin{ex}

 For $\mathfrak{k} \oplus \mathfrak{p}$ sub-Finsler manifolds, the strict abnormal extremals
are never optimal \cite{ABB,B08}. If we write the distribution at a point $x \in G$ and $\mathcal{D}_x=x\mathfrak{p}$,
the Hamiltonian system given by the Pontryagin's maximum principle is integrable concerning elementary functions.
In addition, we have the following expression for geodesics parametrized by arclength, starting at
time zero from $x_0$,
 \begin{equation}\label{Geo11}
   x(t)= x_0 e^{(A_{k}+A_{p})t} e^{-(A_{k})t}
 \end{equation}
such that $A_k \in \mathfrak{k}, A_p \in \mathfrak{p}$, as well as, $\langle A_p, A_p \rangle= \|A_p\|^2=1$.
This is a well known formula in the community \cite{ABB, B08, B14, RuSt15}. Note that the controls whose corresponding trajectories starting from $x_0$ are the normal Pontryagin extremals.

Based on Example \ref{in}, $\mu$ and $\nu$ are complex, the manifold of $SU(2)$ is the sphere $\mathbb{S}^3$ in $\mathbb{R}^4$.
So, we have
$$SU(2)\simeq \mathbb{S}^3= \Bigg\{\begin{pmatrix} \mu \\  \nu \end{pmatrix} \in \mathbb{C}^2\Big|\ |\mu |^2 + |\nu|^2 =1 \Bigg\},$$
by using the isomorphism
$$\eta : SU(2)\to  \mathbb{S}^3, \quad \begin{pmatrix} \mu & \nu\\ -\bar{\nu} & \bar{\mu} \end{pmatrix} \mapsto \begin{pmatrix} \mu \\  \nu \end{pmatrix}.$$
Now, by utilizing formula (\ref{Geo11}), we calculate the explicit expression of geodesics as follows:

First, consider an initial covector in $\mathfrak{su}(2)$ as $A_k \in \mathfrak{k}, A_p \in \mathfrak{p}$, and
 $$\lambda= \lambda(\theta, c)  = \cos(\theta)\sigma_x + \sin(\theta)\sigma_y + c\sigma_z, \qquad \theta \in \mathbb{S}^1, c \in \mathbb{R}.$$
The corresponding exponential map  for all $t \in \mathbb{R}$ is
\begin{equation*}
\mathrm{Exp}(\theta, c, t):= \mathrm{Exp}(\lambda(\theta, c), t) = e^{(\cos(\theta)\sigma_x + \sin(\theta)\sigma_y + c\sigma_z)t}e^{-(c\sigma_z )t}= \begin{pmatrix} \mu \\  \nu \end{pmatrix}
\end{equation*}
such that
\begin{equation*}
\begin{split}
\mu  &= \frac{c\sin(\frac{ct}{2})\sin(\sqrt{1+c^2}\frac{ct}{2})}{\sqrt{1+c^2}}+\cos(\frac{ct}{2})\cos(\sqrt{1+c^2}\frac{t}{2})\\
 & + i\bigg(\frac{c\cos(\frac{ct}{2})\sin(\sqrt{1+c^2}\frac{t}{2})}{\sqrt{1+c^2}}
- \sin(\frac{ct}{2})\cos (\sqrt{1+c^2}\frac{t}{2})\bigg),\\
\nu & = \frac{\sin(\sqrt{1+c^2}\frac{t}{2})}{\sqrt{1+c^2}} \bigg( \cos(\frac{ct}{2}+\theta)+i\sin(\frac{ct}{2}+\theta)\bigg).
\end{split}
\end{equation*}
To compute the explicit expression of geodesics for the shortened
Lorentz group $SO_0(2, 1)$, we choose an initial covector $\mathfrak{so}(2, 1)$ as $A_k \in \mathfrak{k}, A_p \in \mathfrak{p}$, and
 $$\kappa= \kappa(\theta, c)  = \cos(\theta)\sigma_x + \sin(\theta)\sigma_y - c\sigma_z, \qquad \theta \in \mathbb{S}^1, c \in \mathbb{R}.$$
Using formula (\ref{Geo11}), the corresponding exponential map  for all $t \in \mathbb{R}$ is
\begin{equation*}
\mathrm{Exp}(\theta, c, t):= \mathrm{Exp}(\kappa(\theta, c), t) = e^{(\cos(\theta)\sigma_x + \sin(\theta)\sigma_y - c\sigma_z)t}e^{(c\sigma_z )t}
= \begin{pmatrix} \mu \\  \nu \end{pmatrix}.
\end{equation*}
where $\mu$ and $\nu$ are given above.
Obviously, the geodesics for both groups, $SU(2)$ and the shortened
Lorentz group $SO_0(2, 1)$, are the same.
 \end{ex}

Let us now take a look at case of the controllability of linear systems with unbounded controls,
given by the form
\begin{equation}\label{T9}
  \frac{dU}{dt}= AU+Bu,
\end{equation}

such that $U$ is a vector in $\mathbb{R}^n$ (not a unitary propagator), $A$ and $B$ are skew Hermitian matrices, and $u$ is a vector amplitude
in such a way we can rewrite it as $u_jb_j$, so the system becomes
$$\frac{dU}{dt}= AU+\sum_{j}^{m}u_jb_j \in AU + \mathrm{span}\{b_1,\ldots, b_m\},$$
where $u_j$ number multiplying by the columns $b_j$ of the matrix $B$.
This system is said to be a completely controllable if there exists
an unconstrained control law that can steer the system $U$ from the identity to a target.
Even if the system has a drift, i.e. the system is controllable,
it takes arbitrary small
time to steer the system $U$ between points of interest. The solution for such a system
$$U(t)=e^{At}U(0)+\int e^{A(t-s)} B(s)u(s) ds,$$
where $U(0)=I$ is the vector at no time (time zero).

We will be interested in the Lie algebra $\mathfrak{u}(n)$ of skew-Hermitian $n \times n$ matrices considered as
a Lie algebra over the real field. For example, all matrices $- i( \sum_{j=1}^{m} u_j(t)H_j)$ in $(\ref{T41})$ are
in  $\mathfrak{u}(n)$. The subalgebra  $\mathfrak{su}(n)$ of  $\mathfrak{u}(n)$ will play an essential role. It includes
the matrices in  $\mathfrak{u}(n)$ with zero trace.

According to the above setting, we have the ability to module these controls,
i.e. in the sense we can make them positive and negative values.
Thus, it gives us the possibility  to go anywhere exactly.
In fact, we want to find out what is the minimum time takes to do this,
to simplify, the minimum time will only depend on $H_0$ and the target state.

\begin{corollary} \label{11}
  The quantum systems given above, namely, (\ref{T41}),(\ref{T10}) and (\ref{T9}) have the following distributions

  \begin{align}\label{T12}
    \mathcal{D} & = \mathrm{span}\{iH_1U, \ldots , iH_mU\}; \\
    \mathcal{D}_f & = -iH_0U + \mathrm{span}\{iH_1U, \ldots , iH_mU\}; \\
    \mathcal{D}_f & = AU + \mathrm{span}\{b_1,\ldots, b_m\},
  \end{align}
   respectively, where $\mathcal{D}$  is the distribution and $\mathcal{D}_f$ are the affine distributions.
\end{corollary}

One of the most significant roles of the sub-Finsler geometry being in understanding
the optimal trajectories as sub-Finsler geodesics, precisely, it can suitably be called sub-Randers geodesics.
Now comes the importance of the previous Corollary \ref{11}, through which we can easily construct a sub-Finslerian metric
(consequently, a sub-Finsler geometry)
on the distribution
$$ \mathcal{D} = \mathrm{span}\{iH_1U, \ldots , iH_mU\}$$
 that produce from the system (\ref{T41}), (see for instance, \cite{LA}).
Russell et al \cite{RuSt15} studied the H\"{o}rmander's condition and
showed that by using the fact that $\sum_{j=1}^{m} u_j(t)H_j$ is constrained such that
\begin{equation}\label{T13}
  h \big(i\sum_{j=1}^{m} u_j(t)H_j, i\sum_{j=1}^{m} u_j(t)H_j\big) =1, \quad \ \forall t,
\end{equation}
for some inner product $h: \mathfrak{su} \times \mathfrak{su} \to \mathbb{R}$ on $\mathfrak{su}(n)$,
the time optimal trajectories of the time evolution operator $U(t)$ are exactly the geodesics of the following right invariant Randers metric on $SU(n)$,
i.e, the time independent is the constraint and valid for all time. In addition, using some constraints
(e.g. (\ref{T13}))
the optimal trajectories for $U(t)$ are geodesics of Randers metric restricted to an
affine distribution, i.e. $\mathcal{D}_f$ on $SU(n)$. The affine distribution $\mathcal{D}_f \subset TSU(n)$
consisting of vectors of the expression (Corollary (\ref{T13}))
$$\mathcal{D}_f= -iH_0U + \mathrm{span}\{iH_1U, \ldots , iH_mU\},$$
such that $\{iH_1U, \ldots , iH_mU\}\subset \mathfrak{su}(n)$ span the subset of $\mathfrak{su}(n)$
which is $h$-orthogonal to the span of the subset of $\mathfrak{su}(n)$ spanning the 'forbidden
directions'. In the sense that $\mathcal{D}_U=\mathcal{D}U$ the distribution $\mathcal{D}$ is right invariant,
i.e. the optimal trajectories are the length minimizing geodesics that joining the given endpoints
according to a sub-Finsler metric in the sense of restrictions distribution $\mathcal{D}$, and which are parallel to the distribution $\mathcal{D}$.
Such geodesic $V(t)$ is parallel to the distribution $\mathcal{D}$ means that $\frac{dV(t)}{dt}=\mathcal{D}_{V(t)}, \ \forall t.$
Under consideration that the distribution $\mathcal{D}$ is bracket generating,
we can say the system is controllable,
that is, every unitary gate could be implemented. This provides an exact condition
for controllability in
the presence of additional constraints.
Khaneja et al in \cite{KBG01} consider
a quantum system with a drift (\ref{T10}),
left invariant on some compact group $G$.  However, all vector fields $H_j$
belong to $\mathfrak{k}$, a
Lie subalgebra of the Lie algebra $\mathfrak{g}$ of $G$.
Moreover, they chose a Cartan decomposition $\mathfrak{g} = \mathfrak{k} \oplus \mathfrak{p},$
with the standard Cartan’s commutation relations, and so, the Lie algebra generated by the $H_j$'s, $j > 1$
 is not equal to $\mathfrak{g}$ (it is only $\mathfrak{k}$).
 Therefore, in their case, to move from a point in a coset $K_0$
 to another point in a coset $K_1$ requires the use of the drift $H_0$ and hence requires a bounded
speed.
This implies that, even for unbounded controls there is a minimum time,
which is strictly larger than zero (and not attained in general), while Boscain et al in \cite{B02},
show that if we relax the constraint $u_1^2+u_2^2\leq1$ then the minimum time is zero (also not attained in
general).

Now, if we consider $G$ as the unitary group. In the system given by (\ref{T4}),
we denoted by $\mathcal{R}(I,t)$ the {\em reachable set} of all $\hat{U}\in G$ that can be achieved from identity $U(0)=I$
at the time $t$. Moreover,
 the time given by $$t^*(U_f)=\mathrm{inf}\{t\geq 0| \ U_f \in \overline{ \mathcal{R}}(I,t)\},$$
$$t^*(KU_f)=\mathrm{inf}\{t\geq 0| \ kU_f \in \overline{ \mathcal{R}}(I,t), k \in K\},$$
 is said to be {\em infimizing time}, for
producing the propagator $U_f \in G$,
where $\overline{ \mathcal{R}}(I,t)$ is the closure of the set $\mathcal{R}(I,t)$.
Note that the
reachable set from the origin is only the origin. In addition, we define the following notation
$$\mathcal{R}(I ,T):=\underset{0\leqslant t\leqslant T}{\cup}\mathcal{R}(I ,t),$$
$$\mathcal{R}(I):=\underset{0\leqslant t\leqslant \infty}{\cup}\mathcal{R}(I ,t).$$
\begin{rem}
  A system  (\ref{T4}) is said to be controllable if $\mathcal{R}(I, t) = G$  for all  $I \in G$. However, the system is accessible but not controllable
  if the system semigroup $\mathcal{S} := \mathcal{R}(I)$ has an interior point in $G$.
\end{rem}
   The following example show that system (\ref{T4}) is an accessible but not controllable:
\begin{ex}
  Assume that $M=\mathbb{R}^2$ such that the system given by
  $$\begin{bmatrix} \dot{p}_1 \\  \dot{p}_2 \end{bmatrix}=\begin{bmatrix} q^2 \\ 0 \end{bmatrix}+u \begin{bmatrix} 0 \\  1 \end{bmatrix}.$$
  The reachable set from a point $r \in \mathbb{R}^2$ is
  $$\mathcal{R}(r)=\{s \in \mathbb{R}^2\ \mathrm{such \ that} \ s_1 > r_1\} \cup \{r\}.$$ therefore,
  the system is not controllable from the initial point $(0,0)$, however accessible from the same point because
  the point not in the interior.
\end{ex}

Let us define the adjoint system given by
\begin{equation}\label{XS}
   \frac{dU}{dt}= - i\left(X_d+ \sum_{j=1}^{m} u_j(t)X_j,\right)U,\qquad X \in \mathrm{Ad}_KX_d,
\end{equation}
where $U(0)=I$ and $\mathrm{Ad}_K(X_d)$ denotes the adjoint orbit of $X_d \in \mathfrak{g}$ under $K$
 induced by conjugation as follows
$$\mathrm{Ad}_K(X_d)=\{\mathrm{Ad}_{k}(X_d)=k X_dk^{-1}, \ \mathrm{such \ that}\ k\in K\} \subseteq \mathfrak{g}.$$

On the homogeneous sub-Finsler manifold $G/K$,
 let us define a
right invariant control
system for any $P \in G$
\begin{equation}\label{XP}
  \dot{P} = XP, \qquad P(0)=U(0), \ X\in \mathrm{Ad}_KX_d,
\end{equation}
we call such a control system an {\em adjoint
control system}.
\begin{theorem}
  The infimizing time $t^*(U_f)$ for steering the system (\ref{XS}) from  $U(0)=I$ to $U_f$
  is the same as the minimum coset time needed for steering the adjoint system
\begin{equation}\label{ACS}
  \dot{P}=XP, \ \mathrm{where} \ P\in G \ \mathrm{and} \ X \in Ad_K(-X_d),
\end{equation}
from $P(0)=I$ to $KU_f$ \cite{KBG01}.
\end{theorem}

\begin{rem}
  Let $U_f \in G / K$ be an arbitrary target state,
we can replace $U_f$ by any other point
$$kU_fk^{-1} =  kU_f \in  G/K, k \in  K$$
in a manner, the control problem does not change essentially.
Since the set of control variables,
i.e. $ \mathrm{Ad}_K(H_d) $ is invariant under conjugation in $K$ therefore
the optimal trajectory after such a conjugation remains optimal.
In addition, following the general  properties of symmetric
spaces, we have
$$\mathrm{Ad}_KU_f= \{ kU_f \in G/K\ \mathrm{where} \ k \in  K\}$$
has nonempty intersection with $T \subseteq G/K$ such that  $T$ is a
maximal torus in $G$. To be more specific, the following equality holds

$$\mathrm{Ad}_KU_f \cap T = \exp(\mathfrak{W}_X), \qquad \forall X \in \mathfrak{h}, \ \mathrm{with}\  \exp(X)\in \mathrm{Ad}_KU_f.$$
\end{rem}

\section{Conclusion}
The system (\ref{T4}) guarantees that all admissible motions in $G/K$ are generated by the drift
Hamiltonian $H_d$. To pick out an initial direction in $G/K$ it enough to choose a convenient $k \in K$ by means of the controls.
Practically speaking, for each point $KU_{\gamma} \in
G/K$ there is a proper subspace $\mathcal{D}_{KU_{\gamma}}$ of the directions generated by $H_d$ can be
chosen by the controls. Recall that, thanks to the Chow's theorem, we ensure that
every point in $G/K$ can be reached from any given starting point by a piecewise smooth curve determined by the controls since the subspaces $\mathcal{D}_{KU_{\gamma}}$  is bracket generating.
The set of the adjoint
orbit $Ad_K(-iH_d)$ of $-iH_d$ under the action of the subgroup $K$,
in general is not the whole of $\mathfrak{p}$, is the set of all possible directions
in $G/K$,
that means not all the directions in $G/K$ are
directly accessible, unless we achieved that by a back and forth motion in directions we can directly access.
To produce new directions of motion we can apply the noncommuting
generators that mentioned in (\ref{NG}).

This kind of problem, where one has to find the shortest curve between its endpoints in a manifold subject to the constraint that the tangent to the curve always belong to a subset of all possible directions, have
been well studied in the seance sub-Riemannian and sub-Finslerian geometries which are known as (sub-Riemannian) sub-Finslerian geodesics \cite{L.K, L.K1, LoMa00, Mo02}.
If the set of accessible directions is the set $Ad_K(-iH_d)$, then the
the problem of finding time optimal control laws minimize to find sub-Finsler geodesics in coset $G/K$ \cite{KGB02}.
Khaneja et al in \cite{KBG01} and \cite{KGB02}, show that the sub-Riemannian geodesics were computed
for the space $SU(4)/SU(2)\otimes SU(2),$ in the framework of
optimal control of coupled two-spin and three-spin systems, in particular,
if $n = 2$
spins-$1/2$ then $G/K$ takes the form of a Riemannian
symmetric space. Thus time-optimal trajectories between points in $G$ correspond to Riemannian geodesics, but
if $n > 2$ the coset $G/K$ are no longer Riemannian symmetric spaces, so
the time-optimal trajectories in $G$ denote sub-Riemannian geodesics.


\begin{thebibliography}{99}
\bibitem{ABB}
 Agrachev A, Barilari D and  Boscain U 2019
 {\it A Comprehensive Introduction to Sub-Riemannian geometry} (UK:Cambridge Studies in Advanced Mathematics)

\bibitem{L.K}
Alabdulsada L M and Kozma L 2019
On the connection of sub-Finslerian geometry
{\em Int. J. Geom. Methods Mod. Phys.} {\bf 16}(supp02) 1941006

\bibitem{L.K1}
Alabdulsada L M and Kozma L 2020
Hopf-Rinow theorem of sub-Finslerian geometry  submitted

\bibitem{LA}
Alabdulsada L M 2021
Sub-Finsler geometry and nonholonomic mechanics  submitted


\bibitem{BCS00}
Bao D, Chern S and Shen Z 2000
 {\it An Introduction to Riemann-Finsler geometry} (New York: Springer-Verlag) Graduate Texts in
Mathematics 200

\bibitem{BH99}
Brody D C and  Hughston L P 2001
Geometric quantum mechanics
{\it J.Geom.Phys.} {\bf 38} 19-53

\bibitem{B02}
Boscain U, Charlot G, Gauthier J -P, Gu\'{e}rin S and Jauslin H -R 2002
Optimal control in laser-induced population transfer
for two- and three-level quantum systems
{\it  J. Math. Phys.} {\bf 43} 2107-2132

\bibitem{B08}
Boscain U and Rossi F 2008
Invariant Carnot-Caratheodory Metrics on $S^3$, $SO(3)$, $SL(2)$, and Lens Spaces
{\em SIAM J. Control Optim.} {\bf 47}(4)

\bibitem{B14}
Boscain U, Gauthier J -P, Rossi F and Sigalotti M 2014
Controllability of quantum mechanical systems: from conical
eigenvalue intersections to Lie bracket conditions
{\it  21st Int. Symposium on Math. Theory of Networks and Systems} 1892-1895


\bibitem{VB89}
Berestovskii V N 1989 Homogeneous manifolds with an intrinsic metric II
{\it Siberian Math. J.} {\bf 30}(2)  180-191

\bibitem{Chow39}
Chow W -L 1039 \"Uber systeme von linearen partiellen
differentialgleichungen erster ordnung (German) {\em Math. Ann.} {\bf 117} 98-105

\bibitem{JCG}
Clelland J N, Moseley C G and  Wilkens G R 2009
Geometry of control-affine systems,
{\em SIGMA Symmetry Integrability Geom. Methods Appl.} {\bf 5}  095

\bibitem{Al07}
  D'Alessandro D 2007
    {\it Introduction to Quantum Control and Dynamics}
 (New York: CRC Press, Boca, Raton, FL,)

 \bibitem{De12}
 Deng S 2012 {\it Homogeneous Finsler Spaces}  (New York: Springer)

 \bibitem{DoOt16}
  Le Donne E and Ottazzi A 2016
 Isometries of Carnot groups and sub-Finsler homogeneous manifolds
{\it J. Geom. Anal.} {\bf 26}(1) 330–345

 \bibitem{Do}
Le Donne E
 {\it Lecture notes on sub-Riemannian geometry} https://sites.google.com/site/enricoledonne/

 \bibitem{He01}
  Helgason S 2001 Corrected reprint of the 1978 original
 {\it Differential geometry, Lie groups, and symmetric spaces}
  (New York: American Mathematical Society, Providence, RI,)

\bibitem{LoMa00}
 L\'opez C  and  Mart\'\i nez E 2000
Sub-Finslerian metric associated to an optimal control system
  {\em SIAM J. Control Optim.} {\bf 39}(3) 798-811

\bibitem{Mo02}
 Montgomery R 2002 {\it A tour of subriemannian geometries, their geodesics and applications} (New York: American Mathematical Society, Providence, RI,)


\bibitem{KBG01}
Khaneja N, Brockett R and  Glaser S J 2001
Time optimal control in spin systems,
{\it Phys. Rev. A} {\bf 63}, 032308


 \bibitem{KGB02}
   Khaneja N, Glaser S J  and Brockett R  2002
  Sub-Riemannian geometry and time optimal control of three spin systems: Quantum gates and coherence transfer
  {\it Phys. Rev. A} {\bf 65}(3) part A 032301
  Errata {\it Phys. Rev. A} {\bf 68} 049903 (2003);
  {\it Phys. Rev. A} {\bf71}(3) part B 039906 (2005)

 \bibitem{Ko73}
    Kostant B 1973 On convexity, the Weyl group and the Iwasawa decomposition
    {\it Ann. Sci. Ecole Norm. Sup.} {\bf 6} 413-455

    \bibitem{PaSi02}
  Pappas G J and   Simic S  2002
   Consistent hierarchies of affine nonlinear systems {\em IEEE Transactions on Automatic Control} {\bf 47}(5) 745–756

  \bibitem{RuSt15}
      Russell B  and  Stepney S 2015 Zermelo navigation in the quantum brachistochrone,
    {\em J. Phys. A - Math. Theor.}  {\bf 48}  115303


\end{thebibliography}
\end{document}